\documentclass[lettersize,journal]{IEEEtran}
\usepackage{amsmath,amsfonts}
\usepackage{algorithmic}
\usepackage{algorithm}
\usepackage{array}
\usepackage{textcomp}
\usepackage{stfloats}
\usepackage{url}
\usepackage{verbatim}
\usepackage{graphicx}
\usepackage{cite}
\usepackage{amsthm}
\usepackage{booktabs,tabularx}
\usepackage[normalem]{ulem}
\usepackage{xcolor}
\usepackage{multirow}
\usepackage{subfigure}

\newtheorem{theorem}{Theorem}[section]

\newtheorem{definition}{Definition}
\newtheorem{example}{Example}
\newcommand{\LCOMMENT}[1]{\item[] \textit{#1}}

\begin{document}

\title{Redundancy-Driven Top-$k$ Functional Dependency Discovery}

\author{Xiaolong~Wan, Xixian~Han
\thanks{The authors are with School of Computer Science and Technology, Harbin Institute of Technology, China. (e-mail: wxl@hit.edu.cn, hanxx@hit.edu.cn)}
\thanks{Manuscript received XX XX, 202X; revised XX XX, 202X.}}



\maketitle

\begin{abstract}

Functional dependencies (FDs) are basic constraints in relational databases and are used for many data management tasks. Most FD discovery algorithms find all valid dependencies, but this causes two problems. First, the computational cost is prohibitive: computational complexity grows quadratically with the number of tuples and exponentially with the number of attributes, making discovery slow on large-scale and high-dimensional data. Second, the result set can be huge, making it hard to identify useful dependencies. We propose SDP (Selective-Discovery-and-Prune), which discovers the top-$k$ FDs ranked by redundancy count. Redundancy count measures how much duplicated information an FD explains and connects directly to storage overhead and update anomalies. SDP uses an upper bound on redundancy to prune the search space. It is proved that this upper bound is monotone: adding attributes refines partitions and thus decreases the bound. Once the bound falls below the top-$k$ threshold, the entire branch can be skipped. We improve SDP with three optimizations: ordering attributes by partition cardinality, using pairwise statistics in a Partition Cardinality Matrix to tighten bounds, and a global scheduler to explore promising branches first. Experiments on over 40 datasets show that SDP is much faster and uses less memory than exhaustive methods.

\end{abstract}


\section{Introduction}

\IEEEPARstart{F}{unctional} dependencies (FDs) are basic constraints in the relational databases~\cite{DBLP:books/mg/SKS20}. An FD $X \rightarrow A$ states that the values of $X$ determine the value of $A$, and for any two tuples, if they agree on $X$ then they must also agree on $A$, where $X$ is the left-hand side (LHS) and $A$ is the right-hand side (RHS). FDs have many practical applications in data management, including schema normalization~\cite{DBLP:journals/is/KohlerL18}, data integration~\cite{DBLP:journals/csur/MasmoudiLKAZ24}, data cleaning~\cite{DBLP:journals/jdiq/BoecklingB25}, and integrity maintenance~\cite{DBLP:journals/vldb/ZhangL25}.

Given their theoretical and practical significance, researchers have developed many algorithms to discover all valid FDs~\cite{DBLP:journals/pacmmod/BleifussPBSN24,DBLP:conf/sigmod/PapenbrockN16,DBLP:journals/pvldb/PapenbrockEMNRZ15,DBLP:journals/tkde/WanHWL24}. But even with these efforts, FD discovery remains difficult to scale, whose computational complexity is quadratic in the number of tuples and exponential in the number of attributes~\cite{DBLP:journals/tkde/LiuLLC12}. As a result, FD discovery is computationally prohibitive on large-scale and high-dimensional data~\cite{DBLP:journals/tcs/BlasiusFS22}. More importantly, even when computation is feasible, the \emph{result set size} itself often becomes a bottleneck. The number of valid FDs can grow exponentially with respect to the number of attributes. For instance, on the \textit{Flight} dataset with 1,000 tuples and 109 attributes, over 980,000 FDs are discovered~\cite{DBLP:journals/pvldb/PapenbrockEMNRZ15}. Such a massive output overwhelms users and makes it difficult to identify the truly valuable FDs for downstream tasks.

\begin{figure}[t]
	\centering
	\begin{minipage}{0.45\linewidth}
		\centering
		\scriptsize
		\renewcommand{\arraystretch}{1.1}
		\textbf{(a) Relation Instance $r$}
		\vspace{2pt}
		
		\begin{tabular}{|c|c|c|c|}
			\hline
			\textbf{ID} & \textbf{FlightNo} & \textbf{Origin} & \textbf{Dest} \\ \hline
			$t_1$ & UA57 & SFO & JFK \\
			$t_2$ & UA57 & SFO & JFK \\
			$t_3$ & UA57 & SFO & JFK \\
			$t_4$ & AA20 & LAX & ORD \\
			$t_5$ & AA20 & LAX & ORD \\
			\hline
		\end{tabular}
	\end{minipage}
	\hfill
	\begin{minipage}{0.52\linewidth}
		\centering
		\scriptsize
		\textbf{(b) Redundancy Analysis}
		\vspace{5pt}
		
		\textbf{FD 1: ID $\rightarrow$ Origin} \\
		\textit{Grouping by LHS (ID):} \\
		$\{t_1\}, \{t_2\}, \{t_3\}, \{t_4\}, \{t_5\}$ \\
		\textbf{Redundancy:} 0 (No repeated values)
		\vspace{8pt}
		
		\textbf{FD 2: FlightNo $\rightarrow$ Origin} \\
		\textit{Grouping by LHS (FlightNo):} \\
		$\{t_1,t_2,t_3\}$ \quad $\xrightarrow{yields}$ \textbf{2} redundant \\
		$\{t_4,t_5\}$ \quad\quad $\xrightarrow{yields}$ \textbf{1} redundant \\
		\textbf{Total Redundancy:} $2 + 1 = \mathbf{3}$
	\end{minipage}
	\vspace{-2pt}
	\caption{A motivating example for redundancy of FDs}
	\label{fig:motivating_example}
\end{figure}

This motivates moving beyond exhaustive FD discovery toward top-$k$ paradigm that ranks dependencies by their practical impact. A natural criterion is \emph{data redundancy}. For an FD $X \rightarrow A$, the \textit{redundancy count} measures how many repeated values in $A$ are implied by grouping tuples on $X$, and thus how much duplicated information the dependency explains~\cite{DBLP:journals/acta/Vincent99,DBLP:journals/is/WeiL23}. In contrast to treating FDs purely as integrity constraints, redundancy count connects an FD to concrete consequences such as storage overhead and update anomalies~\cite{DBLP:journals/jacm/ArenasL05}. Therefore, prioritizing FDs with large redundancy counts directs the discovery process toward dependencies that are more likely to be useful for normalization and data cleaning.

\begin{example}
	Figure~\ref{fig:motivating_example} shows a simple instance $r$ where both ID $\rightarrow$ Origin and FlightNo $\rightarrow$ Origin hold. The first dependency is induced by a unique identifier. Since ID is a key, ID $\rightarrow$ Origin carries a redundancy count of 0 and offers little beyond what the key already implies. In contrast, FlightNo $\rightarrow$ Origin reflects a recurring regularity in the data, i.e., tuples sharing the same flight number also share the same origin. So its redundancy count is 3 in the motivating example. This example illustrates why redundancy is useful for filtering. It highlights dependencies that capture repeated structure likely to matter for normalization and cleaning.
\end{example}

To solve these problems, we focus on \textit{top-$k$ FD discovery}. Instead of finding all valid FDs, we aim to return only the $k$ dependencies with the highest redundancy counts. Existing studies mainly focus on how to perform exhaustive discovery efficiently, or consider redundancy as conceptual ranking metric for relevance criterion~\cite{DBLP:journals/is/WeiL23}. They do not explain how to use redundancy to actually guide the discovery process. Although the broader data profiling literature, such as finding meaningful keys~\cite{DBLP:journals/pvldb/KoehlerL25} or entity-enhancing rules~\cite{DBLP:journals/pacmmod/FanHWX23,DBLP:journals/pacmmod/FanHXZ24}, has increasingly moved toward selective discovery, a method that directly integrates redundancy into the FD discovery process is still missing.

Filling this gap is not just a matter of adding a ranking step. If we first enumerate all valid FDs and then sort them, we end up paying the full cost of exhaustive discovery, which is the very cost we are trying to avoid. The key is to use the redundancy metric during discovery process so that low-potential regions of the lattice can be pruned early. This is not straightforward, because redundancy count is a data-dependent statistic. It is determined by how tuples are grouped under $X$ and therefore cannot be inferred from the schema alone. In practice, computing it reliably requires checking partitions and thus performing expensive validation, which makes it hard to turn redundancy into a cheap pruning test while enumerating candidates.

To address this challenge, we introduce SDP algorithm (\textit{S}elective-\textit{D}iscovery-and-\textit{P}rune), utilizing hypergraph-traversal-based enumeration framework. Rather than enumerating every valid FD, SDP is designed to search for high-redundancy candidates first and cut off unpromising regions early. The main observation is simple. When we extend the LHS by adding more attributes, the partition of tuples becomes finer, i.e., the number of equivalence classes does not decrease, so the redundancy that an FD can possibly achieve cannot increase. This gives a monotone upper bound on the redundancy of FD candidates in any subtree of the set-enumeration lattice. SDP uses this bound together with the current top-$k$ threshold $\tau$ to perform pruning operation. Once upper bound of a node falls below $\tau$, none of its descendants can enter the top-$k$ results, and the entire branch can be skipped. As a result, SDP avoids validating a large fraction of low-potential candidates while still returning the top-$k$ FDs.

We further speed up SDP with three optimizations. First, we order attributes by partition cardinality so that the search reaches high-redundancy candidates early and raises the threshold $\tau$ quickly. Second, we make the estimated upper bound tighter by using pairwise partition statistics stored in a Partition Cardinality Matrix (PCM), which captures interactions that single-attribute counts miss. Third, instead of exploring RHS trees one by one, we use a global scheduler that interleaves them and always advances the currently most promising branch, accelerating the convergence of $\tau$ in practice.

In summary, we make the following contributions.
\begin{itemize}
	\item \textit{Problem.} We study top-$k$ FD discovery under redundancy count, turning a relevance measure into a concrete retrieval task.
	\item \textit{Algorithm.} We propose SDP, an exact top-$k$ algorithm that uses a monotone upper bound to prune the search space within the hypergraph-transversal-based enumeration framework.
	\item \textit{Optimizations.} We add three practical improvements to reduce the number of explored nodes and validations.
	\item \textit{Evaluation.} Experiments on more than 40 real datasets show that SDP is substantially faster and more memory-efficient than exhaustive baselines.
\end{itemize}

The rest of this paper is organized as follows. Section~\ref{sec:preliminary} provides the necessary preliminaries, followed by a survey of related works in Section~\ref{sec:relatedworks}. Section~\ref{sec:ba} develops the baseline algorithm. The core SDP algorithm and its three optimization strategies are introduced in Section~\ref{sec:sdp} and Section~\ref{sec:optimization}, respectively. Section~\ref{sec:performance} presents a comprehensive experimental evaluation. Finally, Section~\ref{sec:conclusion} concludes the paper.

\section{Preliminaries} \label{sec:preliminary}

This section formalizes the concepts of functional dependencies and redundancy count, and subsequently defines the top-$k$ discovery problem addressed in this paper.

Let $R = \{A_1, A_2, \ldots, A_m\}$ be a relation schema consisting of a set of attributes. A relation instance $r$ over $R$ is a finite set of $n$ tuples, where each tuple is an ordered list of values corresponding to attributes in $R$. For any tuple $t \in r$ and attribute $A \in R$, $t[A]$ denotes the value of $A$ in $t$. Similarly, for a subset of attributes $X \subseteq R$, $t[X]$ denotes the projection of $t$ onto $X$.

\subsection{Functional Dependencies and Redundancy}

A functional dependency (FD) captures a constraint between attribute sets.

\begin{definition}[Functional Dependency]
	Given attribute set $X \subseteq R$ and an attribute $A \in R$, a functional dependency, denoted as $X \rightarrow A$, holds in a relation $r$ if and only if for all pairs of tuples $t_1, t_2 \in r$, $t_1[X] = t_2[X]$ implies $t_1[A] = t_2[A]$.
\end{definition}

Here, $X$ is termed the left-hand side (LHS) and $A$ the right-hand side (RHS). An FD $X \rightarrow A$ is \textit{trivial} if $A \in X$. It is \textit{minimal} if no proper subset $Y \subset X$ satisfies $Y \rightarrow A$. In this work, we focus on identifying non-trivial and minimal FDs.

To quantify the utility of an FD, we utilize the concept of equivalence classes. Given $r$ and $X \subseteq R$, the equivalence class of a tuple $t$ with respect to $X$, denoted by $[t]_X$, is defined as $[t]_X = \{ u \in r \mid u[X] = t[X] \}$. The set of all such equivalence classes constitutes the partition of $r$ by $X$, denoted as $\pi_X = \{ [t]_X \mid t \in r \}$. We use \textit{redundancy count} as the utility objective for top-$k$ FD discovery, and compute it efficiently from partitions induced by the LHS attributes.

\begin{definition}[Redundancy Count]
	For a valid functional dependency $X \rightarrow A$ in relation $r$, the redundancy count, denoted as $red(X \rightarrow A)$, is defined as the summation of redundant occurrences across all equivalence classes:
	\begin{equation*} \label{eq:redundancy}
		red(X \rightarrow A) = \sum_{[t]_X \in \pi_X} (|[t]_X| - 1) = |r| - |\pi_X|.
	\end{equation*}
\end{definition}

Intuitively, within an equivalence class of size $s$, since the LHS values are identical, the FD mandates that the RHS values must also be identical. Thus, $s-1$ values are technically redundant repetitions. A higher redundancy count implies that the FD governs a larger portion of the structural consistency of the data, making it a natural proxy for relevance. Note that $red(X \rightarrow A)$ is determined solely by the LHS partition $\pi_X$ and the total number of tuples $|r|$, provided the FD holds.

\subsection{Problem Statement}

Classical FD discovery focuses on the exhaustive enumeration of all minimal and non-trivial FDs. However, this often yields a huge number of results with varying utility. To address this, we formulate the top-$k$ discovery problem.

\begin{definition}[Top-$k$ FD Discovery]
	Given a relation instance $r$ and a positive integer $k$, let $\mathcal{F}$ be the set of all valid minimal and non-trivial FDs in $r$, top-$k$ FD discovery problem aims to identify a subset $\mathcal{F}_k \subseteq \mathcal{F}$ of size at most $k$, such that $\forall f \in \mathcal{F}_k$ and $\forall f' \in \mathcal{F} \setminus \mathcal{F}_k$: $red(f) \ge red(f')$.
\end{definition}

We strictly consider FDs with positive redundancy. Dependencies with zero redundancy (e.g., key dependencies where $|[t]_X|=1$ for all classes) are excluded as they contribute no data reduction potential. This formulation allows us to focus computational resources on discovering the most impactful knowledge.

\section{Related Work} \label{sec:relatedworks}

\subsection{Functional Dependency Discovery}

The primary objective of classical functional dependency discovery is to identify the \textit{complete} set of minimal and non-trivial FDs that hold on a given relation instance. While existing algorithms differ in their candidate generation and validation mechanisms, they universally adhere to this paradigm of exhaustive discovery. These algorithms can be categorized into three types: attribute-oriented, tuple-oriented, and hybrid approaches.

\textit{Attribute-oriented approaches.} A large body of work discovers FDs by enumerating candidates in the attribute lattice and pruning whenever possible. TANE~\cite{DBLP:journals/cj/HuhtalaKPT99} is the standard level-wise approach, which traverses the lattice level by level, checks candidates using stripped partitions, and removes many candidates using inference rules. FUN~\cite{DBLP:journals/is/NovelliC01} and FD\_Mine~\cite{DBLP:journals/datamine/YaoH08} keep the same enumeration scheme but introduce additional pruning conditions. DFD~\cite{DBLP:conf/cikm/AbedjanSN14} instead separates the search by RHS and explores one lattice per RHS attribute using a randomized depth-first walk. There are also variants that rely on hashing~\cite{DBLP:journals/dke/LiuYLW13}, where tuples with the same LHS values go into the same bucket to speed up FD checks. These approaches work well on moderate schemas, but when $|R|$ is large the number of lattice nodes dominates, and both enumeration and validation become expensive.

\textit{Tuple-oriented approaches.} Tuple-oriented methods take a different route. They derive FDs from agreements and disagreements observed across tuple pairs. Early systems such as FDEP~\cite{DBLP:journals/aicom/FlachS99} and Dep-Miner~\cite{DBLP:conf/edbt/LopesPL00} rely on extensive pairwise comparisons, which quickly becomes expensive as $n$ grows. FastFDs~\cite{DBLP:conf/dawak/WyssGR01} reduces the inspection cost by searching over difference sets in a heuristic depth-first manner, avoiding a full traversal of all pairs. More recently, FSC~\cite{DBLP:journals/tkde/WanHWL24} cuts down comparisons by pre-computing which pairs are comparable, and induces FDs directly from difference sets on those pairs. In practice, tuple-oriented techniques tend to be attractive for wide schemas, but they are less competitive on long tables where the number of tuple pairs dominates.

\textit{Hybrid approaches.} Hybrid methods are motivated by the fact that attribute-based and tuple-based techniques have complementary strengths. HyFD~\cite{DBLP:conf/sigmod/PapenbrockN16} is a representative example. It first samples tuple pairs to generate promising candidates and then switches to an attribute-oriented phase to validate them. The algorithm may switch back and forth between the two phases based on validation efficiency. DHyFD~\cite{DBLP:conf/icde/WeiL19} extends HyFD with an FD-tree variant and dynamic data management. More recently, FDHITS~\cite{DBLP:journals/pacmmod/BleifussPBSN24} casts FD discovery as a hitting-set enumeration search and combines parallel exploration with one-pass validation, which leads to strong performance in modern settings.

\textit{Summary.} Existing FD discovery algorithms all try to find every valid FD, no matter what strategy they use. This can produce huge result sets that overwhelm users and make it hard to identify useful FDs.

\subsection{Top-$k$ Paradigms in Rule Discovery}

Top-$k$ FD discovery solves this by returning only the $k$ most useful dependencies based on a scoring function~\cite{DBLP:journals/debu/Chomicki11, DBLP:journals/tkde/HanLG15, DBLP:journals/tkde/HanLG17}. Instead of overwhelming users with all valid FDs, it identifies the most relevant ones and does not require setting thresholds.

\textit{Rule and Itemset Mining.} Top-$k$ discovery is well-studied 
for association rules, using metrics like support and confidence~\cite{DBLP:journals/dase/WanH24,DBLP:journals/apin/LiuNF21,DBLP:journals/datamine/PetitjeanLTW16}. Similarly, high-utility itemset mining ranks patterns by economic value instead of frequency~\cite{DBLP:journals/ker/KumarS24}. These techniques are designed for transactional data. In this paper, we apply top-$k$ discovery to functional dependencies, which are integrity constraints in relational databases.

\textit{Entity Enhancing Rules.} Top-$k$ discovery has also been applied to Entity Enhancing Rules (REEs)~\cite{DBLP:journals/pacmmod/FanHWX23, DBLP:journals/pacmmod/FanHXZ24}. REEs use active learning to incorporate user feedback on which rules are interesting and diverse, rather than relying only on objective measures. This work shows that ranking methods can work for rules similar to functional dependencies.

\textit{Functional Dependencies.} Wei and Link~\cite{DBLP:journals/is/WeiL23} propose redundancy count to measure the practical utility of FDs. This metric captures how much duplicated information an FD explains, which relates directly to storage overhead and update anomalies. While they propose ranking FDs by redundancy, they do not explain how to use this metric to guide discovery. Existing methods do not integrate redundancy into the search process for top-$k$ results.

\section{Baseline Algorithm} \label{sec:ba}

To the best of our knowledge, no existing algorithm is specifically designed for top-$k$ FD discovery based on redundancy count. Thus, we first establish a baseline method, referred to as FDR (Full-Discovery-and-Rank).

\subsection{FDR algorithm}

FDR uses a two-phase post-processing strategy. It first discovers all valid FDs, then ranks them by redundancy count to find the top-$k$ results.

\subsubsection{Phase 1 (exhaustive FD discovery)}

The first phase discovers all valid FDs, denoted as $\mathcal{F}$, from the relation $r$. We use FDHITS algorithm~\cite{DBLP:journals/pacmmod/BleifussPBSN24}, the state-of-the-art FD discovery algorithm, for a strong baseline.

We detail the internal mechanisms of FDHITS below. Understanding these mechanisms is important because SDP (Section~\ref{sec:sdp}) builds on them to enable efficient top-$k$ discovery.

\textit{Execution Process of FDHITS.} FDHITS extends the hypergraph-based framework of \textsc{HPIvalid}~\cite{DBLP:journals/pvldb/BirnickBFNPS20}. It characterizes the invalidity of an FD as a hitting set problem on \textit{difference sets}. To avoid the quadratic cost $O(n^2)$ of computing all pairwise differences, it employs a \textit{sampling-based approximation} followed by exact validation. The key steps are listed as below.

\begin{itemize}
	\item \textit{Difference set sampling.} Given relation $r$ and schema $R$, the difference set for two tuples $t_1, t_2$ is $\Delta(t_1, t_2) = \{ A \in R \mid t_1[A] \neq t_2[A] \}$. FDHITS avoids generating all pairs by performing weighted sampling on the equivalence class partitions of each attribute. It generates a tentative hypergraph $\mathcal{H}' = (R, \mathcal{D}')$, where $\mathcal{D}'$ contains the difference sets from sampled tuple pairs.
	
	\item \textit{Induced sub-hypergraph.} For a specific RHS attribute $A \in R$, the algorithm constructs an induced sub-hypergraph $\mathcal{H}'_A = (R, \mathcal{D}'_A)$, where $\mathcal{D}'_A = \{ E \setminus \{A\} \mid A \in E, E \in \mathcal{D}'\}$. This structure encodes the constraints necessary for FDs determining $A$.
	
	\item \textit{Minimal hitting set enumeration.} A candidate LHS $X \subseteq R \setminus \{A\}$ is valid with respect to the sample if it hits every edge in $\mathcal{H}'_A$ (i.e., $X \cap E \neq \emptyset, \forall E \in \mathcal{D}'_A$). FDHITS enumerates minimal hitting sets using the MMCS tree search algorithm~\cite{DBLP:journals/dam/MurakamiU14}. \textit{Note that this enumeration is exhaustive and order-agnostic regarding redundancy.}
	
	\item \textit{Validation \& refinement.} Each candidate FD $X \rightarrow A$ derived from the hitting sets is validated against the full dataset $r$ using equivalence class partitions with respect to \textit{X}. If invalid, new difference sets are added to $\mathcal{H}'$, and the search refines.
\end{itemize}

Upon termination, phase 1 outputs the exact set of all valid FDs $\mathcal{F}$.

\subsubsection{Phase 2 (post-processing and ranking)}
In the second phase, FDR iterates through $\mathcal{F}$ to compute the redundancy count. Crucially, during the validation step in phase 1, FDHITS constructs the equivalence class partition $\pi_X$ for each candidate LHS $X$. We leverage this by computing the redundancy count immediately:
\[
red(X \rightarrow A) = \sum_{[t]_X \in \pi_{X}} (|[t]_X| - 1).
\]
Finally, all FDs are sorted in descending order of their redundancy counts, and the top-$k$ FDs are returned.

\subsection{Analysis of the Baseline}

The FDR baseline provides a crucial reference point due to its guaranteed correctness. Since phase 1 is exhaustive, FDR provably returns the exact top-$k$ result. 


\textit{Limitations.} The main problem with FDR is that it separates discovery and ranking. Phase 1 may discover many FDs with very low redundancy (e.g., near-key dependencies), which are then discarded in phase 2. This two-phase approach is very expensive on large datasets. To address this, SDP integrates redundancy into the search process to prune unpromising branches early.

\section{The SDP Algorithm} \label{sec:sdp}

SDP (Selective-Discovery-and-Prune) avoids exhaustive search by integrating redundancy-based pruning into the search process. It builds on the FDHITS framework~\cite{DBLP:journals/pacmmod/BleifussPBSN24}, which uses minimal hitting set enumeration on hypergraphs. The key idea is to use an upper bound on redundancy. Once the upper bound of a candidate falls below the top-$k$ threshold, we skip the entire branch without validation.

\subsection{Upper Bound Estimation Strategy} \label{subsec:sdp:uppest}

To prune effectively, we need an upper bound on redundancy that we can compute quickly. We get this using the monotonicity of partition sizes. Recall that for an FD $X \rightarrow A$, the redundancy count is $red(X \rightarrow A) = n - |\pi_{X}|$. Since $n$ is constant, maximizing redundancy is equivalent to minimizing the partition size $|\pi_X|$.

\begin{theorem} \label{theorem:cardinalityOfSglEqu}
	For any attribute subset $X \subseteq R$ and any attribute $B \in X$, the partition cardinality satisfies:
	\[
	|\pi_B| \;\le\; |\pi_X|.
	\]
\end{theorem}

\begin{proof}
	$\forall t_1, t_2 \in r$, if they belong to the same equivalence class in $\pi_X$ (i.e., $t_1[X] = t_2[X]$), they must imply $t_1[B] = t_2[B]$ since $B \in X$. Thus, every equivalence class in $\pi_X$ is a subset of some equivalence class in $\pi_B$. This implies that $\pi_X$ is a refinement of $\pi_B$, and consequently $|\pi_B| \le |\pi_X|$.
\end{proof}

This monotonicity naturally extends to attribute subsets.

\begin{theorem} \label{theorem:cardinalityOfSubsetEqu}
	For any $Y \subseteq X \subseteq R$, it holds that $|\pi_Y| \le |\pi_X|$.
\end{theorem}
\begin{proof}
	$\forall t_1, t_2 \in r$, if $t_1[X] = t_2[X]$, then $t_1[Y] = t_2[Y]$ for any subset $Y$. Thus, $\pi_X$ refines $\pi_Y$ and $|\pi_Y| \le |\pi_X|$.
\end{proof}

Furthermore, for a valid FD, the RHS imposes a constraint on the partition size.

\begin{theorem} \label{theorem:cardinalityOfRHS}
	For a valid FD $X \rightarrow A$, it holds that $|\pi_A| \le |\pi_X|$.
\end{theorem}
\begin{proof}
	By the definition of an FD, $t_1[X] = t_2[X]$ implies $t_1[A] = t_2[A]$. Thus, tuples in the same LHS equivalence class map to the same RHS value, meaning $\pi_X$ refines $\pi_A$.
\end{proof}

Combining these theorems, we derive the upper bound. Since $|\pi_X| \ge \max_{B \in X} |\pi_B|$ and $|\pi_X| \ge |\pi_A|$ (if valid), the minimum possible size for $|\pi_X|$ is constrained by these lower bounds. Consequently, the maximum possible redundancy count for valid $X \rightarrow A$, denoted $\overline{red}(X \rightarrow A)$, is:

\begin{equation*} \label{eq:upperbound}
	\overline{red}(X \rightarrow A) = \min\left( n - \max_{B \in X} |\pi_B|, \quad n - |\pi_A| \right)
\end{equation*}

\textit{Computational Efficiency.} We pre-compute $|\pi_B|$ for every attribute $B \in R$ during the initial sampling phase. Thus, computing the upper bound for any candidate $X$ only needs to iterate over attributes in $X$, which takes $O(|X|)$ time. This is much faster than validation, so pruning checking is very fast.

\subsection{Execution Process of SDP} \label{sec:sdp:executionprocess}

SDP follows the general pipeline of difference set sampling and hypergraph construction. However, unlike the baseline, it incorporates dynamic pruning. To facilitate this, SDP maintains:
\begin{itemize}
	\item A min-heap $MH$ of size $k$ to store the top-$k$ FDs found so far.
	\item A dynamic pruning threshold $\tau$, defined as the minimum redundancy count in $MH$ (or 0 if $|MH| < k$).
\end{itemize}

For a specific $A \in R$, SDP enumerates the minimal hitting sets on the induced sub-hypergraph $\mathcal{H}_A'$ and aims to identify valid FDs with RHS attribute $A$. Different from FDR, SDP uses the following pruning principle to reduce the search space.

\textit{The Pruning Principle.}
Based on the derived upper bound, we establish the following pruning rule:
\begin{quote}
	\textit{For any candidate LHS $X$ and RHS $A$, if $\overline{red}(X \rightarrow A) < \tau$, then neither $X \rightarrow A$ nor any of its supersets (specializations) can be a top-$k$ result.}
\end{quote}

\begin{theorem} \label{theorem:MonotonicityofUpperBound}
	$\forall X \subseteq R$ and $B \in R \setminus (X \cup \{ A \})$, $\overline{red}(X \cup \{B\} \to A) \le \overline{red}(X \to A)$.
\end{theorem}
\begin{proof}
Let $M(X)=\max_{B\in X}|\pi_B|$. Since $M(X\cup\{B\})=\max\big(M(X),\,|\pi_B|\big)\ge M(X)$, we have $n - M(X\cup\{B\}) \le n - M(X)$.
The term $n-|\pi_A|$ is independent of $X$. For any real numbers $u \le v$ and any constant $c$, it holds that $\min(u,c) \le \min(v,c)$.
Taking $u = n - M(X\cup\{B\})$, $v = n - M(X)$, and $c = n - |\pi_A|$, we obtain $\overline{red}(X \cup \{B\} \rightarrow A) = \min\left( n - M(X\cup\{B\}),\; n - |\pi_A| \right) \le \min\left( n - M(X),\; n - |\pi_A| \right) = \overline{red}(X \rightarrow A)$.
\end{proof}

Theorem \ref{theorem:MonotonicityofUpperBound} proves the monotonicity of upper bound of $red(X \rightarrow A)$ and the correctness of pruning principle.

\textit{Pruning-integrated tree search.}
We perform minimal hitting set enumeration using a set-enumeration tree search adapted from MMCS~\cite{DBLP:journals/dam/MurakamiU14}. Each node $o$ in the tree maintains the following information.
\begin{itemize}
	\item $o.atts$: the current LHS attributes of the candidate.
	\item $o.cand$: potential attributes to extend the candidate.
	\item $o.uncov$: hyperedges in $\mathcal{H}'_A$ not yet covered by $o.atts$, used for identifying the minimal hitting sets.
\end{itemize}

Crucially, we integrate pruning operation into two stages of tree search (see Algorithm \ref{alg:sdpSglAtt}):
\begin{enumerate}
	\item \textit{pre-expansion pruning (Line \ref{line:prune_child}).} Before branching a child node which extends \textit{o} by adding $B$ ($B \in o.cand$), we check if extending the current LHS is promising. If $\overline{red}(o.atts \cup \{B\} \rightarrow A) < \tau$, the child is never pushed to the stack.
	\item \textit{pre-validation pruning (Line \ref{line:prevalidationpruning}).} When a minimal hitting set is found, we check its upper bound again before invoking the expensive partition validation. If $\overline{red}(o.atts \rightarrow A) < \tau$, $o.atts \rightarrow A$ and its specializations cannot be top-$k$ results even though they are valid. Therefore, $o.atts \rightarrow A$ does not need to be validated and the sub-tree rooted at $o$ can be pruned entirely. Otherwise, SDP updates $MH$ if $o.atts \rightarrow A$ is validated to be valid and update $\tau$ (line \ref{line:alg1:9} to line \ref{line:alg1:12}).
\end{enumerate}

\begin{algorithm}[t]
	\renewcommand{\algorithmicrequire}{\textbf{Input:}}
	\renewcommand{\algorithmicensure}{\textbf{Output:}}
	\caption{SDPForSglRHS}
	\label{alg:sdpSglAtt}
	\footnotesize
	\begin{algorithmic}[1]
		\REQUIRE Hypergraph $\mathcal{H}'_A$, target $k$, global heap $MH$.
		\ENSURE Updated heap $MH$.
		\STATE Initialize stack $Q$ with root node $rt$.
		\STATE $\tau \gets (MH.size() < k) ? 0 : MH.min().redundancy$
		
		\WHILE{$Q$ is not empty}
		\STATE $o \gets Q.pop()$; \quad $X \gets o.atts$
		
		\IF{$o.uncov = \emptyset$ \textbf{and} $o$ is minimal} \label{line:case1}
		\LCOMMENT{// Case 1: Candidate found, check bound first.}
		\IF{$\overline{red}(X \to A) \ge \tau$} \label{line:prevalidationpruning}
		\STATE $valid \gets \textsc{Validate}(X \to A, r)$ \COMMENT{Expensive check}
		\IF{$valid$}
		\STATE $rc \gets n - |\pi_X|$ \label{line:alg1:9}
		\IF{$rc > \tau$}
		\STATE $MH.pop()$ and $MH.add(X \to A)$ 
		\STATE $\tau \gets MH.min().redundancy$
		\ENDIF \label{line:alg1:12}
		\ELSE
		\STATE Update $\mathcal{H}'_A$ with new edges from validation failure.
		\STATE Re-evaluate $o.uncov$ with new edges.
		\STATE \textbf{goto} Line \ref{line:expand} \COMMENT{Treat as Case 2 now}
		\ENDIF
		\ENDIF
		\ELSE
		\LCOMMENT{// Case 2: Expand node}
		\STATE $e_m \gets \arg\min_{e \in o.uncov} | e \cap o.cand |$ \label{line:expand}
		\FOR{each $B \in e_m \cap o.cand$} \label{line:alg1:line22}
		\STATE $cld \gets \textsc{GenerateChild}(o, B)$
		\IF{$\overline{red}(cld.atts \to A) \ge \tau$} \label{line:prune_child}
		\STATE $Q.push(cld)$ \COMMENT{Pruning branches}
		\ENDIF
		\ENDFOR
		\ENDIF
		\ENDWHILE
		\RETURN $MH$
	\end{algorithmic}
\end{algorithm}

\begin{algorithm}[t]
	\caption{SDP Algorithm (Overall)}
	\label{alg:sdpOverall}
	\footnotesize
	\begin{algorithmic}[1]
		\REQUIRE Relation $r$, integer $k$.
		\ENSURE Top-$k$ FDs.
		\STATE Compute partitions $\pi_A$ for all $A \in R$.
		\STATE $\mathcal{H}' \gets \textsc{DifferenceSetSampling}(r)$
		\STATE $MH \gets \text{new MinHeap}(k)$
		
		\FOR{each attribute $A \in R$} \label{line:alg2:line4}
		\STATE $\mathcal{H}'_A \gets \textsc{InduceHypergraph}(\mathcal{H}', A)$
		\STATE $MH \gets \textsc{SDPForSglRHS}(\mathcal{H}'_A, k, MH)$
		\ENDFOR
		\RETURN $MH$
	\end{algorithmic}
\end{algorithm}

Algorithm \ref{alg:sdpOverall} outlines the global procedure of SDP. By maintaining a global heap $MH$ across iterations for different RHS attributes, the threshold $\tau$ rises monotonically.  This means that FDs discovered early for one attribute help prune the search space for subsequent attributes, making the algorithm faster over time. As proved in Theorem \ref{theorem:OptimalityofSDP}, when SDP terminates, the FDs in $MH$ are top-$k$ results $\mathcal{F}_k$.

\begin{theorem} \label{theorem:OptimalityofSDP}
	The set $\mathcal{F}_k$ identified by the SDP algorithm satisfies the top-$k$ discovery requirement: $\forall f \in \mathcal{F}_k$ and $\forall f' \in \mathcal{F} \setminus \mathcal{F}_k, red(f) \ge red(f')$.
\end{theorem}

\begin{proof}
	According to Theorem \ref{theorem:MonotonicityofUpperBound}, the upper bound $\overline{red}$ is monotonically non-increasing as attributes are added to the LHS. This implies that for any FD $f'$ in the subtree rooted at the current candidate $X \to A$, its true redundancy count is bounded by $\overline{red}(X \to A)$. If the pruning condition $\overline{red}(X \to A) < \tau$ is met, it follows that $red(f') < \tau$, and any $f'$ in the pruned branch cannot possibly belong to the global top-$k$ set. Therefore, no potential top-$k$ results are omitted, ensuring the global optimality of $\mathcal{F}_k$.
\end{proof}

\section{Algorithm Optimization Strategies} \label{sec:optimization}

The basic SDP algorithm works correctly but can be slow. Its speed depends on two factors: how fast the threshold $\tau$ increases and how tight the upper bound is. We optimize both factors with three improvements.

\textit{Speed up $\tau$.} We make $\tau$ increase faster by searching high-redundancy parts of the lattice first. This uses attribute ordering (Section~\ref{sec:opt:order}) and global guidance (Section~\ref{sec:opt:global}).

\textit{Tighter bounds.} We compute a tighter upper bound $\overline{red}^*$ using pairwise attribute information from a partition cardinality matrix (Section~\ref{sec:opt:attpair}).

\subsection{Heuristic Attribute Access Order} \label{sec:opt:order}

Pruning works better when we search high-redundancy candidates first. A random order does not make $\tau$ increase quickly. We use a global ordering $\sigma$ to search high-redundancy parts of the lattice first.

\textit{The Heuristic.} We sort all attributes in ascending order of their partition sizes. Let $\sigma$ be a permutation of $\{1, \ldots, m\}$ such that:
\[
|\pi_{A_{\sigma(1)}}| \le |\pi_{A_{\sigma(2)}}| \le \cdots \le |\pi_{A_{\sigma(m)}}|
\]

We use this ordering in two ways to find high-redundancy FDs faster. Since $\overline{red}(X \to A) = \min(n - \max_{B \in X} |\pi_B|, n - |\pi_A|)$, we want to keep both terms large: choose small $|\pi_A|$ for the second term, and add small $|\pi_B|$ so the first term stays high.

\begin{itemize}
	\item \textit{RHS selection (outer loop, line \ref{line:alg2:line4} in Algorithm \ref{alg:sdpOverall}).} We iterate through RHS 	attributes $A$ in order $\sigma$. By choosing $A$ with small $|\pi_A|$ first, 	the term $(n - |\pi_A|)$ is maximized. Since valid FDs must satisfy $|\pi_X| \ge |\pi_A|$, attributes with small partitions are more likely to be the RHS of high-redundancy FDs. Small $|\pi_A|$ makes this constraint easier to satisfy, making high-redundancy FDs more likely.
	
	\item \textit{LHS extension (inner loop, line \ref{line:alg1:line22} in Algorithm \ref{alg:sdpSglAtt}).} Within the MMCS tree, we 	extend candidate LHS $X$ by adding attributes $B$ also in order $\sigma$.	Adding attributes with small $|\pi_B|$ minimizes the increase in $\max_{B \in X} |\pi_B|$, keeping the upper bound high. This makes the search favor LHS combinations with few equivalence classes, which lead to high redundancy.
\end{itemize}

By searching high-bound candidates first, $\tau$ increases faster, which prunes more branches for the rest of the search. Note that this ordering is purely heuristic and does not affect the correctness of SDP.

\subsection{Tighter Pruning with Attribute Pairs} \label{sec:opt:attpair}

The single-attribute upper bound can be loose. For instance, two attributes $B$ and $C$ might individually have small partitions, but their combination $BC$ could form a key (large partition). To capture such interactions, we incorporate partition cardinalities of attribute pairs.

\textit{The pairwise bound.} Using the property $|\pi_{BC}| \ge \max(|\pi_B|, |\pi_C|)$, we refine the upper bound for a candidate FD $X \rightarrow A$:
\begin{equation*}
	\overline{red}^*(X \rightarrow A) = n - \max \left(|\pi_A|, \max_{B \in X} |\pi_B|, \; \max_{BC \subseteq X} |\pi_{BC}| \right).
\end{equation*}
Clearly, $\overline{red}^* \le \overline{red}$, offering strictly stronger pruning power.

\textit{Partition cardinality matrix} (PCM). Computing all $\binom{m}{2}$ pairwise partitions is prohibitively expensive ($O(m^2 n)$). We address this with a partition cardinality matrix (PCM), which uses a hybrid strategy.
\begin{enumerate}
	\item \textit{Selective pre-computation.} We pre-compute pairwise cardinalities only for the first $d$ attributes in the heuristic order $\sigma$. Setting $d = \lceil \sqrt{m} \rceil$ ensures the pre-computation cost remains comparable to a single pass over all attributes, satisfying the budget constraint.
	\item \textit{Opportunistic updates.} During the search, whenever a candidate FD with a 2-attribute LHS (e.g., $BC \to A$) is validated, $|\pi_{BC}|$ is computed as a byproduct. We cache this value in the PCM. This allows the PCM to learn and densify as the algorithm progresses, improving pruning for future branches at zero extra cost.
\end{enumerate}

\textit{Incremental bound computation.} 
A naive evaluation of $\overline{red}^*$ takes $O(|X|^2)$. We optimize this to $O(|X|)$ by using the set-enumeration tree structure. Let a child node $o_{child}$ (attribute set $Y$) be generated from parent $o_{parent}$ (attribute set $X$) by adding attribute $C$ (i.e., $Y = X \cup \{ C \}$). We maintain a variable ${M}^*(Y)$ tracking the maximum partition size observed within $Y$ (considering both single attributes and available pairs). It is updated incrementally:
\[
M^*(Y) = \max \left( M^*(X), \; |\pi_C|, \; \max_{B \in X} \text{PCM}(B, C) \right),
\]
where $\text{PCM}(B, C)$ returns $|\pi_{BC}|$ if available in PCM, or 0 otherwise. This requires iterating only over the existing attributes in $X$, ensuring the check remains linear $O(|X|)$. 

The corresponding upper bound $	\overline{red}^*(X \rightarrow A) = n - \max(M^*(X), |\pi_A|)$, and $\overline{red}^*(X \rightarrow A) \le \overline{red}(X \rightarrow A)$. 

\begin{theorem}
\label{thm:monotonicity}
For any two attribute sets $X$ and $Y$ where $X \subseteq Y$, and a fixed RHS $A$, it holds that $\overline{red}^*(Y \rightarrow A) \le \overline{red}^*(X \rightarrow A)$.
\end{theorem}

\begin{proof}
Recall that $M^*(X)$ represents the maximum partition cardinality among single attributes and available pairs within $X$. Since $X \subseteq Y$, the set of components considered for $M^*(X)$ is a subset of those for $M^*(Y)$, implying $M^*(X) \le M^*(Y)$. Consequently, $\max(M^*(X), |\pi_A|) \le \max(M^*(Y), |\pi_A|)$. Subtracting from $n$ reverses the inequality, yielding the stated result.
\end{proof}

Theorem \ref{thm:monotonicity} guarantees that the pruning strategy is safe. Specifically, if a node is pruned because $\overline{red}^* < \tau$, its descendants can be skipped as their redundancy bounds cannot exceed $\overline{red}^*$. Since $\overline{red}^* \le \overline{red}$ always holds, SDP inherits the correctness of the base framework while achieving strictly stronger pruning power. 

In the following, we use $\overline{red}^*$ as the \textit{default upper bound} in SDP to maximize pruning efficiency.

\subsection{Accelerated Convergence via Global Best-First Search} \label{sec:opt:global}

Processing RHS attributes one by one may waste time on an RHS with no top-$k$ FDs while other RHS attributes have better candidates. To fix this, we use a global best-first search strategy.

Instead of processing the $m$ set-enumeration trees sequentially, we schedule them using a global priority queue. We maintain $m$ independent stacks $Q_1, \ldots, Q_m$, where $Q_i$ tracks the DFS traversal for RHS $A_i$. A global priority queue $\mathcal{P}$ manages the search frontiers (i.e., top nodes) in $Q_1, \ldots, Q_m$. Each element in $\mathcal{P}$ is a tuple $(score, i)$, where $i$ is the stack index and $score$ is $\overline{red}^*$ of the top node in $Q_i$.

\textit{Execution Flow.}
The scheduler works as follows.
\begin{enumerate}
	\item Extract the tuple with highest score from $\mathcal{P}$. Let this tuple be $(\overline{red}^*, i)$, representing the best active branch across all RHS.	
	\item Pop and process the top node from $Q_i$ (either validate the candidate FD, or expand the node by pushing its children to $Q_i$).	
	\item If $Q_i$ is not empty, compute $\overline{red}^*$ for the new top node and re-insert the updated tuple into $\mathcal{P}$.
\end{enumerate}

This ensures the algorithm always work on the best candidate across all RHS attributes. High-redundancy FDs are found earlier, making $\tau$ increase faster and pruning more branches across all trees.

\begin{theorem}
	\label{theorem:independence}
	Given an upper bound $\overline{red}^*$ and a dynamically updated threshold $\tau$, any traversal strategy that visits all candidate nodes $X$ where $\overline{red}^*(X \rightarrow A) \ge \tau$ guarantees to return the correct top-$k$ FD set.
\end{theorem}

\begin{proof}
	Let $\mathcal{F}_k$ be the true set of top-$k$ FDs, and let $\tau_{final}$ be the $k$-th highest redundancy count in $\mathcal{F}_{k}$. Since $\tau$ is monotonically non-decreasing during the process, $\forall f \in \mathcal{F}_{k}$, its score satisfies $red(f) \ge \tau_{final} \ge \tau_{current}$ at any time step. By the definition of a safe upper bound, $\overline{red}^*(f) \ge red(f) \ge \tau_{current}$. Consequently, $f$ will never be pruned regardless of the traversal order. Since the strategy guarantees visiting all unpruned nodes, every $f \in \mathcal{F}_{k}$ is eventually validated and included in the result.
\end{proof}

It should be noted that the global best-first search uses more memory due to $m$ search tree maintained, but the speedup is worth it. Since $\tau$ increases faster, SDP can prune more branches early, reducing the total search space. As proved in Theorem \ref{theorem:independence}, global scheduler only changes the traversal order without affecting correctness. The final top-$k$ set is the same as the sequential approach.

\section{Experimental evaluation} \label{sec:performance}

This section evaluates SDP on real datasets with varying dimensions and sizes. The experiments compare SDP with the baseline FDR and include an ablation study to measure the contribution of each optimization. This section first describes the experimental setup, then presents and analyzes the experimental results.

\subsection{Experimantal Setup} \label{sec:exp:setup}

We implement the algorithms in Java (jdk-21\_windows-x64). The experiments are run on a DELL OptiPlex 7010MT workstation configured with Intel(R) Core(TM) i9-13900 CPU @ 2.00GHz (24 cores), 64 GB of main memory, a 4 TB hard disk, and the Windows 11 (64-bit) operating system.

The real-world data sets are collected from the UCI Machine Learning Repository\footnote{http://archive.ics.uci.edu/} and Kaggle\footnote{https://www.kaggle.com/}. For a fair comparison, we allocate 56 GB memory (out of 64 GB total) to each algorithm using the -Xmx parameter.

\textit{NULL semantics}. In real-world data sets, null values are widely used to represent missing or unavailable entries. The interpretation of null values impacts the discovery results. In our experiments, we use the standard NULL-EQ semantics, treating all null values as equal but distinct from any non-null value \cite{DBLP:journals/pvldb/Berti-EquilleHN18,DBLP:journals/pvldb/BirnickBFNPS20}.

We evaluate algorithm performance using five metrics.

\textit{Runtime}. We record the total wall-clock time from the start of discovery until completion on each dataset.

\textit{Peak memory consumption}. We measure the peak memory consumption of SDP and FDR using the MemoryMXBean interface of JVM.

\textit{Search Pruning Ratio (PR$_{S}$)}. This measures pruning in the set-enumeration tree. Let $\#$Nodes(SDP) be the number of nodes visited by SDP and $\#$TotalNodes(FDR) be the number of nodes visited by the exhaustive baseline, and $\mathrm{PR}_S = 1 -\frac{\text{$\#$Nodes(SDP)}}{\text{$\#$TotalNodes(FDR)}}$. It reflects the fraction of search space avoided by SDP.

\textit{Candidate Pruning Ratio ($PR_C$)}. This measures pruning at the candidate validation level. Let $\#$AllCandidates(FDR) be the number of FD candidates considered by FDR and $\#$VerifiedCandidates(SDP) number of candidates validated by SDP, and $\mathrm{PR}_C = 1 -\frac{\text{$\#$VerifiedCandidates(SDP)}}{\text{$\#$AllCandidates(FDR)}}$. This shows how well SDP filters irrelevant candidates before validation.

\textit{Result Utility (LIFT$_k$).} LIFT$_k$ measures the quality of top-k FDs as the ratio of average redundancy of top-$k$ FDs to that of remaining FDs ($\bar{s}_{\text{top}}/\bar{s}_{\text{rest}}$), where $\bar{s}_{\text{top}}=\frac{1}{k}\sum_{f \in \mathcal{F}_k}red(f)$ and $\bar{s}_{\text{rest}}=\frac{1}{|\mathcal{F} \setminus \mathcal{F}_k|}\sum_{f \in \mathcal{F} \setminus \mathcal{F}_k} red(f)$. LIFT$_k > 1$ means the top-$k$ FDs have much higher average redundancy than the remaining FDs.

\begin{table*}[t]
	\centering
	\caption{Experimental results on real-life data sets.}
	\footnotesize
	\setlength{\tabcolsep}{4pt} 
	\begin{tabular}{|c||c|c|c|c||c c|c c||c|c|c|c|}
		\hline
		Dataset & $|r|$ & $|R|$ & \#FDs & \#FDs$_{>0}$ 
		& \multicolumn{2}{c|}{FDR} 
		& \multicolumn{2}{c||}{SDP} 
		& Speedup & PR$_S$ & PR$_{C}$ & LIFT$_k$ \\
		\cline{6-9}
		& & & & 
		& Time[s] & Mem[G] 
		& Time[s] & Mem[G] 
		& & & & \\
		\hline
		ad\_table & 268255 & 24 & 3019 & 2996 & 56.054 & 15.8 & 8.575 & 3.5 & 6.537 & 0.831 & 0.872 & 7.511 \\
		
		adult & 32561 & 15 & 78 & 78 & 0.314 & 0.093 & 0.284 & 0.093 & 1.106 & 0.0196 & 0.0672 & 18.897 \\
		
		airbnb\_open\_data & 102599 & 26 & 20149 & 12064 & 76.96 & 16.615 & 38.257 & 8.062 & 2.012 & 0.835 & 0.850 & 10.299 \\
		
		airline\_dataset & 98619 & 15 & 864 & 164 & 1.623 & 2.344 & 1.44 & 2.313 & 1.127 & 0.594 & 0.773 & 199.646 \\		
		
		airlines\_flights & 300153 & 12 & 43 & 31 & 1.282 & 0.584 & 0.995 & 0.474 & 1.288 & 0.266 & 0.451 & 3.616 \\
		
		aggnycyellowtaxi & 921373 & 18 & 2261 & 2189 & 140.529 & 5.375 & 118.771 & 8.563 & 1.183 & 0.348 & 0.414 & 18.009 \\
		
		atom & 160000 & 31 & 7273 & 5135 & 43.173 & 5.5 & 0.978 & 0.519 & 44.144 & 0.968 & 0.990 & 11.401 \\
		
		avila & 10430 & 11 & 336 & 79 & 0.156 & 0.093 & 0.146 & 0.093 & 1.069 & 0.676 & 0.817 & 118.439 \\
		
		bank\_trans & 1048567 & 9 & 59 & 26 & 4.695 & 1.125 & 3.701 & 1.438 & 1.269 & 0.339 & 0.443 & 361.391 \\		
		
		blerssi & 5191 & 15 & 14 & 1 & 0.095 & 0.005 & 0.046 & 0.0052 & 2.065 & 0.086 & 0.6 & $\infty$ \\
		
		carbon\_nanotubes & 10721 & 8 & 72 & 24 & 0.115 & 0.0313 & 0.078 & 0.0625 & 1.474 & 0.407 & 0.542 & 14.944 \\
		
		census & 199523 & 42 & - & - & -- & OM & 21.213 & 6.362 & -- & -- & -- & -- \\
		
		conflongdemo\_jsi & 164860 & 8 & 51 & 5 & 0.268 & 0.158 & 0.188 & 0.156 & 1.426 & 0.323 & 0.522 & $\infty$ \\
		
		covtype & 581012 & 55 & -- & -- & -- & OM & 53.27 & 24.104 & -- & -- & -- & -- \\
		
		creditcard\_2023 & 568630 & 31 & 814 & 784 & 3.129 & 3.967 & 2.771 & 4.031 & 1.129 & 0.901 & 0.985 & 1.039 \\
		
		ditag\_feature & 3960124 & 13 & 58 & 38 & 17.534 & 8.5 & 12.56 & 7.767 & 1.396 & 0.4 & 0.658 & 16.559 \\
		
		eshopclothing & 165474 & 14 & 22 & 14 & 0.582 & 0.178 & 0.448 & 0.151 & 1.299 & 0.243 & 0.386 & 5.124 \\		
		
		fd-reduced-30 & 250000 & 30 & 89571 & 23 & 49.26 & 5.094 & 20.459 & 4.937 & 2.407 & 0.948 & 0.958 & 7027.556 \\
		
		flight & 1000 & 109 & 982631 & 734583 & 195.179 & 2.5 & 0.166 & 0.093 & 1175.8 & 0.9999 & 0.99999 & 21.167 \\
		
		ftda & 13094522 & 11 & 57 & 47 & 536.898 & 25.843 & 68.204 & 18.625 & 7.872 & 0.137 & 0.341 & 172.586 \\
		
		gendergap & 4746 & 21 & 5520 & 3162 & 0.866 & 1.156 & 0.447 & 0.25 & 1.937 & 0.771 & 0.809 & 202.858 \\
		
		global\_health\_stat & 1000000 & 22 & 91658 & 3480 & 886.744 & 5.875 & 864.045 & 10.656 & 1.026 & 0.701 & 0.839 & 128.968 \\
		
		grocery\_sales & 6758125 & 9 & 17 & 3 & 9.788 & 7.906 & 9.795 & 8.134 & 0.999 & 0.363 & 0.5 & $\infty$ \\		
		
		har70plus & 2259597 & 8 & 9 & 2 & 6.013 & 1.581 & 4.372 & 1.596 & 1.375 & 0.25 & 0.478 & $\infty$ \\
		
		harth & 5603043 & 8 & 58 & 18 & 12.26 & 8.005 & 12.808 & 9.531 & 0.957 & 0.196 & 0.396 & 47043.7 \\
		
		healthtracking & 365000 & 12 & 96 & 2 & 0.945 & 0.649 & 0.607 & 0.546 & 1.556 & 0.64 & 0.779 & $\infty$ \\
		
		homec & 503911 & 32 & -- & -- & -- & OM & 140.398 & 18.046 & -- & -- & -- & -- \\
		
		horse & 300 & 28 & 128726 & 128565 & 23.494 & 1.375 & 5.608 & 0.718 & 4.189 & 0.735 & 0.744 & 13.087 \\
		
		hr\_mnc\_data & 2000000 & 13 & 251 & 32 & 20.382 & 18.781 & 20.179 & 14.125 & 1.01 & 0.495 & 0.668 & 3208927 \\
		
		ht\_sensor & 928991 & 12 & 949 & 467 & 23.613 & 8.531 & 28.378 & 8.718 & 0.832 & 0.459 & 0.598 & 2.174 \\
		
		internetfirewall & 65532 & 12 & 88 & 88 & 0.754 & 0.156 & 0.251 & 0.0732 & 3.004 & 0.508 & 0.695 & 2.578 \\
		
		lung\_cancer & 890000 & 17 & 2120 & 984 & 102.881 & 5.187 & 77.09 & 8.781 & 1.334 & 0.410 & 0.531 & 15.869 \\
		
		metropt3 & 1516948 & 17 & 345 & 313 & 69.645 & 20.687 & 58.403 & 12.299 & 1.192 & 0.203 & 0.318 & 2.767 \\
		
		miningprocess & 737453 & 24 & 38125 & 38125 & 567.004 & 11.031 & 458.145 & 11.594 & 1.237 & 0.652 & 0.705 & 217.339 \\
		
		movies\_dataset & 999999 & 17 & 4073 & 91 & 31.219 & 12.719 & 23.299 & 13.718 & 1.339 & 0.638 & 0.792 & 333969 \\
		
		openfda\_drug & 1000 & 39 & 7789 & 7502 & 0.94 & 0.312 & 0.084 & 0.0312 & 11.191 & 0.991 & 0.998 & 4.534 \\
		
		reactionnetwork & 65554 & 29 & 42960 & 42954 & 137.882 & 44.888 & 1.669 & 0.764 & 82.614 & 0.983 & 0.991 & 1.347 \\
		
		rideshare\_kaggle & 693071 & 57 & -- & -- & -- & OM & 129.544
		& 36.171 & -- & -- & -- & -- \\
		
		sg\_bioentry & 184292 & 9 & 19 & 4 & 0.407 & 0.175 & 0.278 & 0.181 & 1.464 & 0.667 & 0.75 & $\infty$ \\
		
		superconductivty & 21263 & 82 & -- & -- & -- & OM & 0.952 & 0.2897 & -- & -- & -- & -- \\
		
		tnweather\_1.8m & 1888422 & 14 & 103 & 40 & 11.433 & 10.781 & 9.268 & 6.531 & 1.233 & 0.289 & 0.5026 & 428306 \\		
		
		wecsdataset & 287999 & 49 & 61207 & 61207 & 85.397 & 42.94 & 27.261 & 9.963 & 3.133 & 0.899 & 0.959 & 1.0004 \\
		
		youtube\_rec & 1000000 & 14 & 328 & 124 & 16.286 & 6.593 & 15.414 & 8.593 & 1.056 & 0.186 & 0.473 & 59941 \\
		
		zomato & 51717 & 17 & 618 & 532 & 1.621 & 2.3125 & 1.322 & 2.343 & 1.226 & 0.374 & 0.491 & 17.204 \\		
		
		\hline
	\end{tabular}
	\label{table:overallPerformance}
\end{table*}

\subsection{Overall Performance Comparison}

Table~\ref{table:overallPerformance} summarizes the experimental results of SDP and FDR on over 40 real-world datasets from public repositories. The datasets have 8 to 109 attributes and hundreds to tens of millions of tuples, covering a wide range of sizes.

Each algorithm is executed under the same 56 GB memory constraint. If memory consumption exceeds this limit, the case is marked as OM (OutOfMemoryError). We use $k$ = 10 for all datasets and algorithms, a typical choice for top-$k$ discovery. The sensitivity of different $k$ values is discussed in Section~\ref{sec:exp:sensitivityK}.

SDP is much faster than FDR on high-dimensional and large-scale datasets, achieving 10-1000$\times$ speedup while using less memory. For example, on \textit{flight} and \textit{reactionnetwork}, SDP achieves speedups of 1175$\times$ and 82$\times$, respectively. This is because $\tau$ increases quickly. The global scheduler of SDP finds high-redundancy FDs first, making $\tau$ increase faster than sequential search in FDR. Using the monotonic upper bound, SDP then prunes many branches early. The search pruning ratio (PR$_S$) shows that, it exceeds 0.99 on datasets like \textit{flight} and \textit{openfda\_drug}, meaning SDP visits less than 1$\%$ of the nodes FDR requires.

FDR runs out of memory on five datasets: \textit{census}, \textit{covtype}, \textit{homec}, \textit{rideshare\_kaggle}, and \textit{superconductivity}. This happens because FDR materializes partitions for all candidates, which uses too much memory on high-dimensional data. SDP handles all these datasets by pruning early. It skips candidates before materializing their partitions. The candidate pruning ratio (PR$_C$) often exceeds 0.95 on wide datasets, showing SDP skips most low-redundancy candidates. This keeps memory usage low, which depends on the number of high-redundancy FDs found, not the total schema size.

LIFT$_k$ also shows that redundancy is an effective metric. On some datasets like \textit{conflongdemo\_jsi} and \textit{har70plus}, LIFT$_k$ reaches infinity, meaning only the top-$k$ FDs have non-trivial redundancy while the rest have near-zero redundancy. This shows redundancy effectively distinguishes important FDs from low-redundancy ones. By focusing on redundancy, SDP finds the most useful FDs efficiently.

\textit{Performance overhead in low-complexity scenarios.} SDP achieves substantial speedups on high-dimensional datasets, but comparable to or slightly slower than FDR on some low-dimensional datasets (e.g., \textit{ht\_sensor}, \textit{harth}, \textit{grocery\_sales}). For instance, on \textit{ht\_sensor} ($|R|=12$), SDP is about 4.7 seconds slower than FDR. This happens because the optimizations of SDP (e.g., pre-computing PCM and maintaining the global priority queue) add overhead. When the lattice is small and exhaustive search finishes in seconds, this overhead becomes noticeable. However, this small overhead is acceptable. SDP avoids the failures FDR encounters on complex datasets like \textit{flight} and \textit{superconductivity}.

\begin{figure}
	\centering
	\renewcommand{\thesubfigure}{}
	\subfigure[(a) Execution time]{
		\includegraphics[scale = 0.63]{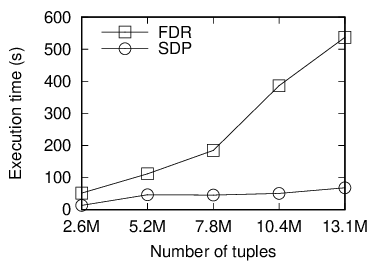}}
	\subfigure[(b) Peak memory usage]{
		\includegraphics[scale = 0.63]{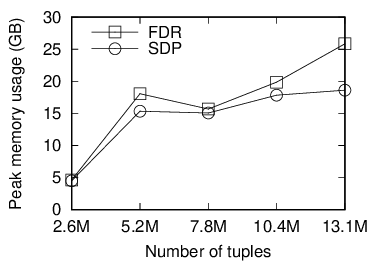}}
	\subfigure[(c) Number of generated nodes]{
		\includegraphics[scale = 0.63]{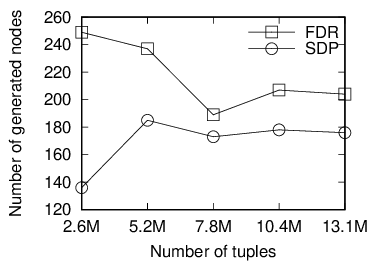}}
	\subfigure[(d) Validation number of FDs]{
	\includegraphics[scale = 0.63]{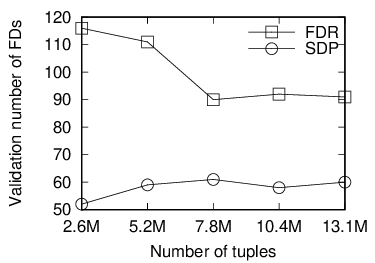}}
	\caption{The illustration of scalability analysis with tuple numbers (ftda data set).}
	\label{fig:scalability_tuples}
\end{figure}

\subsection{Scalability analysis with tuple numbers}

We evaluate the scalability of SDP as the number of tuples increases. We use the largest dataset, \textit{ftda}, which contains over 13 million tuples and 11 attributes. We create prefix subsets of \textit{ftda} at 20$\%$, 40$\%$, 60$\%$, 80$\%$, and 100$\%$ of its original size. This shows how performance degrades as data grows while keeping data characteristics similar. With $k$ = 10, we run both SDP and FDR on each subset and measure their runtime.

Figure \ref{fig:scalability_tuples} shows the scalability of SDP on \textit{ftda} as data grows from 2.6 million to 13.1 million tuples. We analyze the results in three aspects.

\textit{Runtime Efficiency.} Figure \ref{fig:scalability_tuples}(a) shows the runtime of FDR scales linearly with data size, while SDP grows sub-linearly, achieving 7.87$\times$ speedup at 13.1M tuples. This is because our optimizations work together effectively. Attribute ordering and global search make $\tau$ increase quickly, while PCM tightens the upper bound. Together, these reduce the cost of partition refinement. Even though verifying a single candidate costs more as data grows, SDP reduces the number of candidates needing verification, keeping total runtime low.

\textit{Why Node Count of FDR Decreases.} Figures \ref{fig:scalability_tuples}(c) and (d) show an interesting pattern. The node count of FDR decreases as data grows from 2.6M to 7.8M tuples, then stabilizes. This happens because of data distribution convergence. In smaller samples (2.6M), spurious dependencies often appear valid, forcing FDR to explore deeper into the lattice. As data grows to 7.8M, more counter-examples emerge, causing invalid FDs to be pruned earlier in the search tree. In contrast, SDP remains stable, validating a constant 50-60 candidates. This shows the upper bound of SDP works well regardless of data size. Unlike FD validity (which is sensitive to counter-examples), the top-$k$ high-redundancy FDs remain the same as data grows, so SDP maintains consistent pruning.

\textit{Memory Usage.} Figure \ref{fig:scalability_tuples}(b) shows SDP uses less memory than FDR, saving 7 GB at 13.1M tuples. The memory usage of SDP depends on two factors. First, the global search adds some overhead by maintaining execution states (stacks and priority queues) for $m$ concurrent trees. Second, aggressive pruning prevents materializing many intermediate partitions. The results show the memory saved by pruning is much larger than the overhead, so SDP avoids the out-of-memory failures that affect FDR.

\begin{figure}
	\centering
	\renewcommand{\thesubfigure}{}
	\subfigure[(a) Execution time]{
		\includegraphics[scale = 0.63]{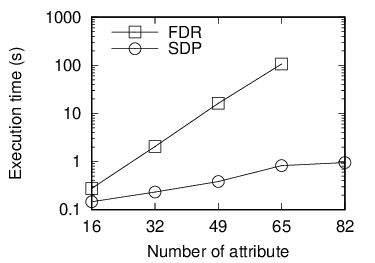}}
	\subfigure[(b) Peak memory usage]{
		\includegraphics[scale = 0.63]{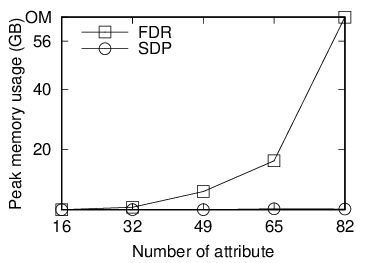}}
	\subfigure[(c) Number of generated nodes]{
		\includegraphics[scale = 0.63]{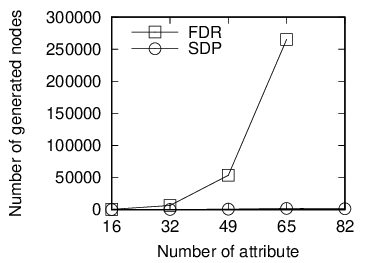}}
	\subfigure[(d) Validation number of FDs]{
		\includegraphics[scale = 0.63]{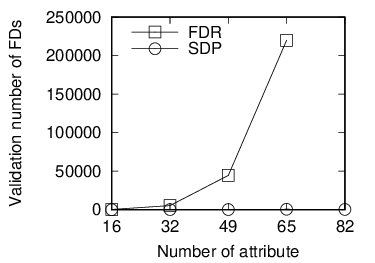}}
	\caption{The illustration of scalability analysis with attribute numbers (superconductivity data set).}
	\label{fig:scalability_attributes}
\end{figure}

\subsection{Scalability analysis with attribute numbers}

We evaluate how SDP scales as the number of attributes increases. This is challenging because the search space grows exponentially with attributes. We use the \textit{superconductivity} dataset (82 attributes), one of our widest datasets. We create subsets containing 20$\%$, 40$\%$, 60$\%$, 80$\%$, and 100$\%$ of the attributes. With $k$ = 10, we run both SDP and FDR on each subset. Figure \ref{fig:scalability_attributes} shows how each algorithm scales with attributes.

\textit{Why FDR Fails.} Figure \ref{fig:scalability_attributes}(a) shows the runtime of FDR grows much faster than SDP as attributes increase from 16 to 82. The runtime of FDR grows exponentially and runs out of memory at 82 attributes. This happens because the search space grows exponentially ($2^{|R|}$). Each added attribute doubles the search space, causing the number of nodes to grow from 419 (at 16 attributes) to 2.65 million (at 65 attributes). In contrast, the runtime of SDP grows much more slowly. This is because its optimizations reduce the impact of combinatorial explosion. Attribute ordering and global search make $\tau$ increase quickly, while PCM tightens the upper bound. Together, these help SDP search only high-redundancy candidates, avoiding the exponential growth of low-redundancy combinations.

\textit{Memory Usage.} Figure \ref{fig:scalability_attributes}(b) shows why FDR fails at 82 attributes. The memory usage of FDR surges from 0.09 GB at 16 attributes to 16.3 GB at 65 attributes, then exceeds the 56 GB limit at 82 attributes. This happens because FDR materializes partitions for all candidates to verify FDs. In contrast, SDP uses much less memory, staying below 0.38 GB. This is because the pruning of SDP determines most candidates cannot enter the top-$k$ results, so SDP never allocates memory for them. This avoids memory failures even with many attributes.

\textit{Why SDP Validates Fewer Candidates.} Figures \ref{fig:scalability_attributes}(c) and (d) show why SDP is faster. As attributes increase from 16 to 65, FDR validates 329 to 219,855 FDs. In contrast, SDP validates only 31 to 246 FDs. Even at 82 attributes, SDP generates 1,518 nodes and validates only 246 candidates. This shows that while the total number of FDs grows exponentially, the number of candidates SDP needs to validate grows slowly. PCM helps SDP achieve this. As dimensions increase, the pairwise information in PCM helps SDP prune candidates earlier. This shows SDP finds the top-$k$ high-redundancy FDs while skipping hundreds of thousands of low-redundancy FDs.

\begin{figure}
	\centering
	\renewcommand{\thesubfigure}{}
	\subfigure[(a) Execution time]{
		\includegraphics[scale = 0.63]{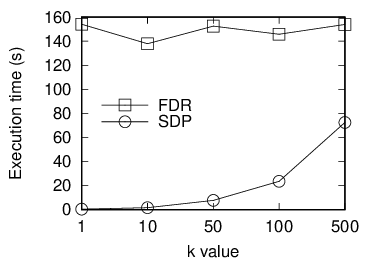}}
	\subfigure[(b) Peak memory usage]{
		\includegraphics[scale = 0.63]{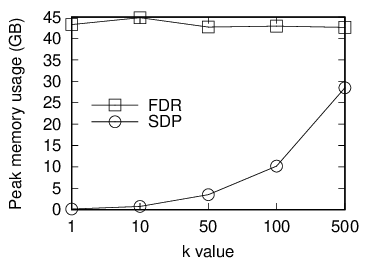}}
	\subfigure[(c) Number of generated nodes]{
		\includegraphics[scale = 0.63]{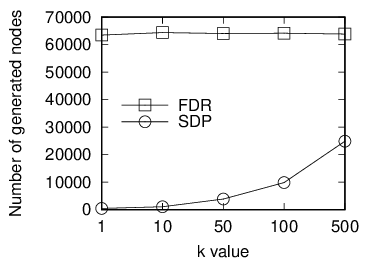}}
	\subfigure[(d) Validation number of FDs]{
		\includegraphics[scale = 0.63]{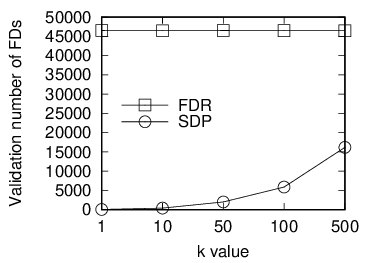}}
	\caption{The illustration of impact of \textit{k} (ReactionNetwork data set).}
	\label{fig:scalability_k}
\end{figure}

\subsection{Parameter Sensitivity to $k$} \label{sec:exp:sensitivityK}

We investigate how the performance of SDP changes with $k$, the number of top FDs to retrieve. We use the \textit{ReactionNetwork} dataset, which contains 42,954 FDs with non-zero redundancy. This allows us to test SDP with different $k$ values. We run SDP with $k\in$$\{$1, 10, 50, 100, 500$\}$, ranging from the single most redundant FD to 500 FDs. We measure runtime for each $k$. Figure \ref{fig:scalability_k} shows SDP and FDR behave differently as $k$ changes.

\textit{Why FDR is Insensitive to $k$.} Figure \ref{fig:scalability_k} shows the performance of FDR is almost flat and insensitive to $k$. The runtime of FDR stays around 145-155 seconds, and memory usage stays between 42.6 GB and 44.9 GB. This is because FDR uses a two-phase approach. It first discovers all valid FDs, then selects the top-$k$ results. So the cost of FDR depends only on the dataset, not $k$. The small variations come from system randomness and hypergraph sampling, not $k$.

\textit{Why SDP is Sensitive to $k$.} In contrast, the performance of SDP increases with $k$. As $k$ increases from 1 to 500, its runtime grows from 0.43s to 72.45s, and the number of generated nodes grows from 442 to 24,000. This happens because $k$ affects $\tau$. The threshold $\tau$ is the redundancy count of the $k$-th best FD found so far. When $k$ is small (e.g., $k$ = 1), $\tau$ quickly reaches a high value, so pruning is very aggressive. SDP prunes almost all candidates, validating only 53. However, as $k$ increases to 500, $\tau$ becomes lower. A lower $\tau$ means more candidates pass the upper bound check, so SDP must explore more candidates, leading to more nodes and validations.

\textit{Memory Usage.} Figure \ref{fig:scalability_k}(b) shows the memory consumption of SDP increases with $k$. Memory usage rises from 0.14 GB ($k$ = 1) to 28.49 GB ($k$ = 500). This is because when $k$ is small, $\tau$ is high, so SDP prunes most candidates before creating their partitions (only 1$\%$ pass). When $k$ is large, $\tau$ is lower, so more candidates pass the upper bound check, and SDP must materialize more partitions. However, even at $k$ = 500, SDP (28.5 GB) uses much less memory than FDR (44 GB), showing that the pruning of SDP is effective even with large $k$.

\begin{table}[htbp]
	\centering
	\caption{Ablation study of optimization components in SDP}
	\setlength{\tabcolsep}{2pt} 
	\begin{tabular}{|c|c|c|c|c|c|c|}
		\hline
		\textbf{Dataset} & \textbf{variant} & \textbf{Time[s]} & \textbf{Mem[G]} & \#TRN & \#VAN & \#FD$_{g}$ \\
		\hline
		\multirow{4}{*}{atom}  %
		& SDP-core & 15.836 & 3.1 & 2746 & 1408 & 1260 \\
		& SDP+O &  4.991 & 3.1 & 534 & 247 & 192 \\
		& SDP+P & 4.912 & 1.6 & 462 & 218 & 181 \\
		& SDP (full) & 0.978 & 0.5 & 379 & 90 & 50 \\
		\hline
		\multirow{4}{*}{census}  
		& SDP-core & - & OM & - & - & - \\
		& SDP+O & 315.055 & 30.9 & 25089 & 17879 & 17655 \\
		& SDP+P & 355.054 & 29.8 & 25447 & 17885 & 17655 \\
		& SDP (full) & 21.213 & 6.4 & 2548 & 1135 & 1066 \\
		\hline
		\multirow{4}{*}{ReactionNetwork}  
		& SDP-core & 81.792 & 32.4 & 27829 & 17736 & 15586 \\
		& SDP+O & 8.239 & 4.6 & 3892 & 2600 & 2293 \\
		& SDP+P & 8.514 & 4.5 & 3791 & 2558 & 2263 \\
		& SDP (full) & 2.849 & 1.3 & 1688 & 724 & 523 \\
		\hline
	\end{tabular}

	\label{tab:ablation_table}
\end{table}

\begin{figure*}
	\centering
	\renewcommand{\thesubfigure}{}
	\subfigure[(a) atom data set]{
		\includegraphics[scale = 0.45]{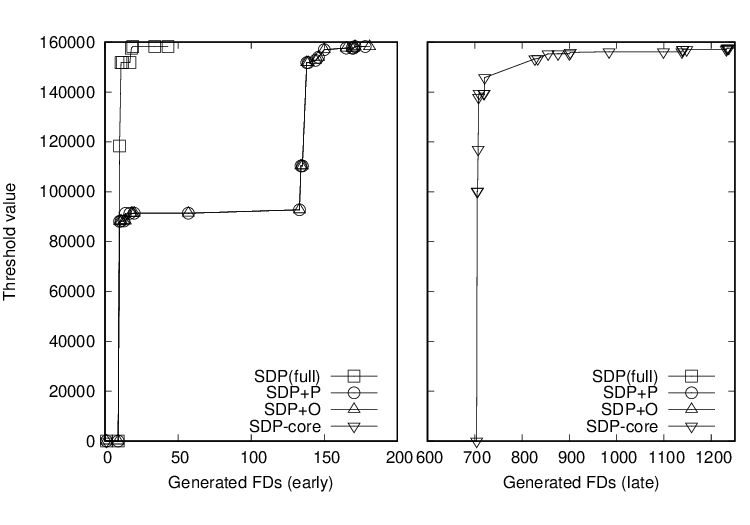}}
	\subfigure[(b) census data set]{
		\includegraphics[scale = 0.45]{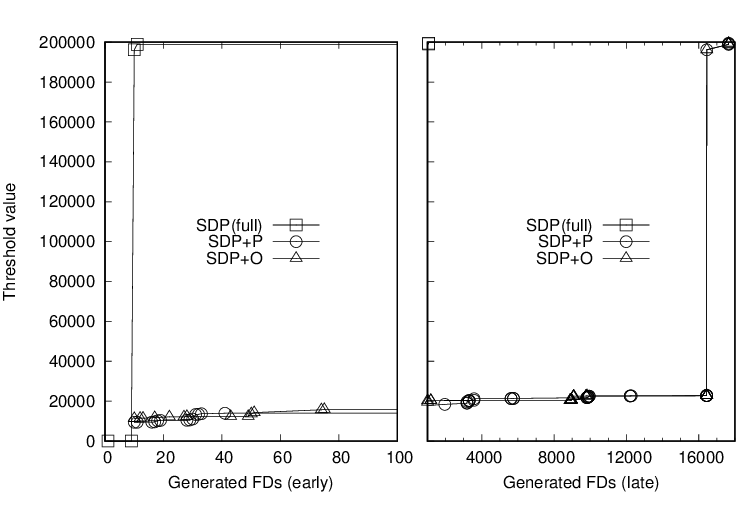}}
	\subfigure[(c) reactionnetwork data set]{
		\includegraphics[scale = 0.45]{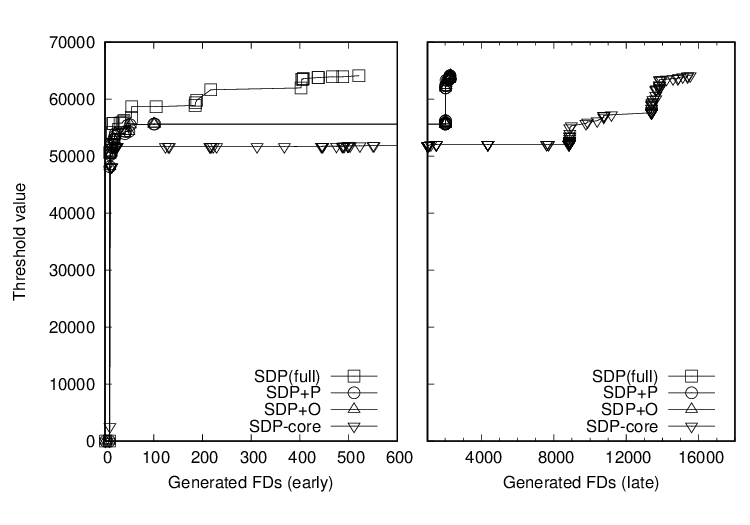}}
	\caption{The illustration of threshold value for different optimization strategies.}
	\label{fig:ablation_curves}
\end{figure*}

\subsection{Ablation study of optimization strategies}
\label{sec:exp:ablation}

We conduct a detailed ablation study to evaluate each of three optimization strategies. To isolate the effect of each component, we compare the performance of all four variants of our algorithm. 

\begin{itemize}
	\item SDP-core: The basic version of SDP, which includes the unified pruning rule but processes RHS attributes in a sequential order and does not include any of the advanced optimizations.
	\item SDP+O: This version enhances SDP-core by adding our first optimization, the heuristic ordering strategy.
	\item SDP+P: This version builds upon SDP+O by also adding our second optimization, the tighter upper bound estimation with pairwise attribute information from the PCM.
	\item SDP (full): The complete version of SDP, which replaces the static RHS processing order with our third optimization, the globally-guided search strategy.
\end{itemize}

We run experiments on three datasets (\textit{atom}, \textit{census}, and \textit{ReactionNetwork}) with $k$ = 10. Table \ref{tab:ablation_table} shows the results. We measure runtime (Time[s]), peak memory (Mem[G]), number of nodes generated (\#TRN), number of validations performed (\#VAN), and FDs generated before finding top-$k$ results (\#FDg).

\subsubsection{Impact of heuristic ordering}

SDP-core and SDP+O show that attribute ordering is critical. Table~\ref{tab:ablation_table} shows SDP-core fails with out-of-memory on \textit{census} and takes 81.792s on \textit{ReactionNetwork}. Heuristic ordering (SDP+O) fixes the memory failure and reduces runtime to 8.239s (10$\times$ speedup).

This is because heuristic ordering helps increase $\tau$ faster. SDP-core processes attributes in random order, so it may explore low-redundancy branches first and $\tau$ increases slowly. A low $\tau$ cannot prune effectively. In contrast, SDP+O prioritizes low-cardinality attributes. This helps find high-redundancy FDs quickly, making $\tau$ increase faster. Figure \ref{fig:ablation_curves}(c) confirms this. The $\tau$ value of SDP+O rises much earlier than that of SDP-core, so SDP+O prunes more aggressively and avoids exploring low-redundancy branches.

\subsubsection{Impact of PCM-guided pruning}

We compare SDP+O and SDP+P to show the impact of PCM. Runtime improvement is small in some cases (e.g., \textit{atom}: from 4.991s to 4.912s), but memory usage drops significantly. On \textit{atom}, peak memory drops by 50$\%$, from 3.1 GB to 1.6 GB.

This is because PCM provides tighter upper bounds. Heuristic ordering increases $\tau$, but the algorithm still needs tight upper bounds to prune effectively. SDP+O uses loose single-attribute bounds. SDP+P uses PCM to compute tighter upper bounds based on pairwise attribute information. So SDP+P detects invalid branches earlier (at smaller $|X|$) and prunes candidates without materializing their partitions. Although computing PCM adds slight overhead (explaining the small runtime increase on \textit{census}), the memory savings show PCM is effective.

\subsubsection{Impact of global search strategy}

SDP (full) adds global search and achieves the largest speedup. On \textit{census}, runtime drops from 355.054s (SDP+P) to 21.213s, and nodes decrease from 25,000 to 2,548. On \textit{atom}, runtime drops from 4.9s to 0.978s.

This is because global search shares $\tau$ across all RHS attributes. In sequential approaches, the $m$ search trees are isolated. A high $\tau$ found in the last tree cannot help earlier trees that have already finished. SDP (full) uses a global priority queue to schedule all $m$ search trees. When any tree finds a high-redundancy FD, $\tau$ updates immediately and prunes all other trees. Figure \ref{fig:ablation_curves}(a) and (b) show the $\tau$ value of SDP (full) reaches its maximum almost immediately, so SDP prunes other attributes early. This reduces node generation significantly.

\begin{table*}[htbp]
	\centering
	\caption{Representative Top-$k$ FDs and Their Semantic Significance Across Diverse Domains}
	\setlength{\tabcolsep}{4pt} 
	\begin{tabular}{|c|c|c|l|c|c|c|}
		\hline
		\textbf{Dataset} & $|r|$ & \textbf{Rank} & \multicolumn{1}{c|}{\textbf{top-k FDs}} & \textbf{\#red} & \textbf{LIFT$_k$} & \textbf{Semantic Significance} \\
		\hline
		movies\_dataset & 999999 & 1 & ReleaseDate $\rightarrow$ ReleaseYear & 972242 & 725750 & Natural temporal hierarchy \\
		\hline
		
		\multirow{2}{*}{tnweather\_1.8m} & \multirow{2}{*}{1888422}
		& 1 & precipitation $\rightarrow$ rain & 1888120 & \multirow{2}{*}{627303} & \multirow{2}{*}{Attribute synonymy equivalence}\\
		& & 2 & rain $\rightarrow$ precipitation & 1888120 &  & \\
		\hline
		
		\multirow{2}{*}{youtube\_rec} & \multirow{2}{*}{1000000} 
		& 1 & \{watch\_time, watch\_percent\} $\rightarrow$ video\_duration & 509623 & \multirow{2}{*}{156585} & \multirow{2}{*}{Calculated mathematical invariants}\\
		& & 2 & \{video\_duration, watch\_time\} $\rightarrow$ watch\_percent & 509623 & & \\
		\hline
		
		\multirow{2}{*}{hr\_mnc\_data} & \multirow{2}{*}{2000000} 
		& 1 & Job\_Title $\rightarrow$ Department & 1999971 & \multirow{2}{*}{9209523} & \multirow{2}{*}{Organizational business policies} \\
		& & 2 & Hire\_Date $\rightarrow$ Status & 1994521 & & \\
		\hline
		
		\multirow{3}{*}{gendergap} & \multirow{3}{*}{4746} 
		& 1 & \{C\_api, C\_man\} $\rightarrow$ gender & 4739 & \multirow{3}{*}{168} & \multirow{3}{*}{Survey metadata indexing logic} \\
		& & 2 & weightIJ $\rightarrow$ E\_NEds & 4730 & & \\
		& & 3 & NIJ $\rightarrow$ E\_NEds & 4730 & & \\
		\hline
	\end{tabular}
	\label{tab:case_study_summary}
\end{table*}

\subsection{Case Study}

We select five representative datasets for qualitative analysis (Table \ref{tab:case_study_summary}). These datasets come from different domains (e.g., spatial, medical, scientific) and have different sizes ($10^3$ to $10^6$ tuples). Redundancy helps identify the most important FDs in each domain.

Due to space constraints, Table \ref{tab:case_study_summary} shows 1-3 representative top-$k$ FDs for each dataset. These examples cover different types of dependencies (e.g., symmetric equivalence, hierarchical structures, business logic).

\textit{1) movies\_dataset}. The top-1 FD (ReleaseDate $\rightarrow$ ReleaseYear) shows a temporal hierarchy. With redundancy count 972,242 and LIFT$_k$ 725,750, this FD shows ReleaseYear is determined by ReleaseDate across nearly the entire dataset. This suggests normalizing the schema to save storage.

\textit{2) tnweather\_1.8m}. The mutual dependencies (precipitation $\leftrightarrow$ rain) both have redundancy count 1,888,120, nearly equal to the dataset size (1,888,422). This shows these columns represent the same information (attribute synonyms). This suggests merging or removing one column to reduce redundancy.

\textit{3) youtube\_rec}. The top-ranked FDs between watch time, duration, and percentage (LIFT$_k$: 156,585) show mathematical relationships in the logging system. These dependencies can be used to validate data integrity. Any violations indicate errors in calculation logic or data ingestion.

\textit{4) hr\_mnc\_data}. This dataset has the highest LIFT$_k$ ($>$9.2 million) in our study. The FDs (e.g., Job$\_$Title $\rightarrow$ Department) show organizational policies: job titles determine departments. These FDs can be used for role-based access control and compliance auditing.

\textit{5) gendergap}. Although this dataset is small (4,746 tuples), the FDs show how survey indices map to gender. The mapping from categorical indices (C$\_$api, C$\_$man) to gender has redundancy 4,739, holding for 99.8$\%$ of tuples. This shows demographic categories are well-defined. The 7 violating tuples (0.2$\%$) can help identify potential errors in demographic fields.

\section{Conclusion} \label{sec:conclusion}

In this paper, we tackle the problem of information overload in functional dependency (FD) discovery. Instead of finding all FDs, we propose SDP (\textit{Selective-Discovery-and-Prune}) to find the top-$k$ FDs. We use redundancy count as a metric to rank FDs. Our key contribution is a tight, monotonic upper bound that enables exact discovery while pruning most of the search space. We also propose three optimizations: heuristic attribute ordering, PCM-guided pruning, and global scheduling that coordinates multiple searches. Experiments on over 40 datasets show SDP achieves 10-1000× speedup on high-dimensional and large-scale data while maintaining exact top-$k$ results. Case studies show the top-$k$ FDs reveal meaningful patterns useful for data management.

\section*{Acknowledgments}
This work was supported in part by National Natural Science Foundation of China grant no. 62402135, U21A20513, Taishan Scholars Program of Shandong Province grant no. tsqn202211091, Shandong Provincial Natural Science Foundation grant no. ZR2023QF059.

\bibliographystyle{IEEEtran}
\bibliography{IEEEexample}




\vfill

\end{document}